\documentclass[11pt,reqno,oneside]{amsart}

\usepackage{amsmath,amssymb,amscd,amsthm}
\usepackage{mathrsfs}
\usepackage{mathtools}
\usepackage{enumitem}
\usepackage{cite}

\usepackage{graphicx}
\usepackage{xcolor}

\usepackage[colorlinks,
	linkcolor={black!30!blue},
	citecolor={black!30!blue},
	urlcolor={black!30!blue}
]{hyperref}

\usepackage{tikz}
\usetikzlibrary{calc,
	matrix,
	fadings,
	fit,
	math,
	arrows.meta}
\usepackage{pgfplots}
\pgfplotsset{compat=1.18}
\usepackage{tikzscale}

\usepackage{bbm}
\usepackage[margin=1.1in]{geometry}

\def\idty{\mathbbm{1}} 

\newcommand{\bbC}{{\mathbb{C}}}

\newcommand{\bbN}{{\mathbb{N}}}

\newcommand{\bbS}{{\mathbb{S}}}
\newcommand{\bbT}{{\mathbb{T}}}
\newcommand{\bbZ}{{\mathbb{Z}}}

\newcommand{\CH}{{\mathcal{H}}}

\newcommand{\SU}{{\mathbb{SU}}}

\newtheorem{theorem}{Theorem}[section]
\newtheorem{lemma}[theorem]{Lemma}

\newtheorem{coro}[theorem]{Corollary}

\newtheorem{remark}[theorem]{Remark}

\theoremstyle{definition}

\allowdisplaybreaks \numberwithin{equation}{section}

\DeclareMathOperator{\dom}{Dom}

\makeatletter
\def\subsection{\@startsection{subsection}{2}%
	\z@{.5\linespacing\@plus.7\linespacing}{.5\linespacing}%
	{\normalfont\scshape\centering}}
\makeatother

\hypersetup{
	pdftitle = {Exponential Suppression of Transport in Electric Quantum Walks},
	pdfauthor = {Houssam Abdul-Rahman,
		Christopher Cedzich,
		Albert H. Werner,
		Günther Stolz},
	pdfsubject = { },
	pdfkeywords = {ballistic transport, periodic field, quantum walk, CMV matrix, velocity	}
}

\begin{document}
	
\title[Exponential Suppression of Transport in Electric Quantum Walks]{Exponential Suppression of Transport\\in Electric Quantum Walks}
		
\author[H.\ Abdul-Rahman]{Houssam Abdul-Rahman}
\address{[H.\ Abdul-Rahman] Department of Mathematical Sciences, United Arab Emirates University, AL Ain, UAE}
\email{\href{mailto:houssam.a@uaeu.ac.ae}{houssam.a@uaeu.ac.ae}}
	
\author[C.\ Cedzich]{Christopher Cedzich}
\address{[C.\ Cedzich] Fakult\"at f\"ur Mathematik und Informatik, FernUniversit\"at in Hagen, Universit\"atsstr. 1, 58097 Hagen, Germany}
\address{Heinrich Heine Universit\"at D\"usseldorf, Universit\"atsstr. 1, 40225 D\"usseldorf, Germany}
\email{\href{mailto:christopher.cedzich@fernuni-hagen.de}{cedzich@fernuni-hagen.de}}

\author[G. Stolz]{Günter Stolz}
\address{[G. Stolz] Department of Mathematics, University of Alabama at Birmingham, Birmingham, AL 35294, USA}
\email{\href{mailto:stolz@uab.edu}{stolz@uab.edu}}

\author[A.~H. Werner]{Albert H. Werner}
\address{[A.~H. Werner], {QMATH}, Department of Mathematical Sciences, University of Copenhagen, Universitetsparken 5, 2100 Copenhagen, Denmark}
\email{\href{mailto:werner@math.ku.dk}{werner@math.ku.dk}}

	
\begin{abstract}
	We establish exact scalings for the maximal group velocity of translation-invariant quantum walks in periodic electric fields. Our main result shows that the maximal group velocity decays exponentially with the period of the field in the whole parameter range, thus affirming a conjecture of \cite{ARS2023-CMP} and at the same time augmenting it to an exact equality. 
	We further demonstrate explicit revival relations and characterize the absolutely continuous spectrum in these models. 
	Our results apply directly also to generalized CMV matrices.
\end{abstract}
	
\maketitle
	
\hypersetup{
	linkcolor={black!30!blue},
	citecolor={black!30!blue},
	urlcolor={black!30!blue}
}

\section{Introduction}

Periodic quantum systems are commonly associated with ballistic transport, where states propagate linearly in time  \cite{aschMotionPeriodicPotentials1998,damanikWhatBallisticTransport2024}. An intriguing and fundamental question arises naturally: is it possible to manipulate the velocity by suitably tuning the system's (periodic) parameters? Addressing this question is the primary theme and driving motivation of the present work.

The physical setup in which we study this question are so-called quantum walks, which have emerged as powerful models for studying quantum dynamics.
Originally introduced as quantum analogues of classical random walks, quantum walks capture essential aspects of quantum coherence and interference. In folklore wisdom, this makes them propagate ballistically compared to the diffusive spread of their classical counterparts \cite{Kempe2003, VenegasAndraca2012}. They have since found numerous applications as subroutines in quantum algorithms \cite{Portugal2013,  Montanaro2016}, as quantum simulators \cite{Childs2009, AspuruGuzik2012}, and in the study of topological phenomena \cite{kitagawaExploringTopologicalPhases2010,TopClass,AsboBB,Asbo} where new invariants arise due to the discrete nature of time. 

From a mathematical perspective, quantum walks offer a rich environment for rigorous spectral and dynamical analysis, connecting features of discrete quantum systems to harmonic analysis \cite{recurrence,bourgainQuantumRecurrenceSubspace2014,WeAreSchur,krushchev}, operator theory \cite{krushchev,OldIndex,F2W,TopClass}, and localization phenomena \cite{HJ14,HJS2009, AhlbrechtVolkherScholzWerner2011,JM10, HamzaJoyeStolz2006, hamzaLyapunovExponentsUnitary2007,schaefer2025dynamicallocalizationtransportproperties,locQuasiPer}. Particularly noteworthy is the connection between so-called split-step quantum walks and CMV matrices \cite{BHJ2003,CGMV2012QIP}, which are five-diagonal unitary matrices that play a fundamental role in the theory of orthogonal polynomials and can be thought of as a unitary counterpoint of Jacobi matrices \cite{canteroFivediagonalMatricesZeros2003,Simon2005OPUC1,Simon2005OPUC2}. 

Quantum walks subject to external fields have been extensively studied \cite{ARS2023-CMP, CGWW2019JMP,CFO1,CFGW2020LMP,Wojcik,Buerschaper}. In the simplest case, electric fields in one spatial dimension, the dynamics are highly sensitive to the field's arithmetic nature \cite{ewalks,locQuasiPer}: For rational electric fields, the system is effectively periodic, which results in ballistic transport at large time scales. Surprisingly, on timescales of the order of the denominator of the field, ``revivals'' occur: the initial state is reproduced with an accuracy that is exponentially good in the period of the system. On the other hand, irrational electric fields produce complex dynamics that depend on their arithmetic properties: well-approximable fields produce hierarchical dynamics while generic fields lead to Anderson localization characterized by pure point spectrum and exponentially localized eigenfunctions \cite{locQuasiPer}.

In this paper, we establish an exact scaling for the maximum group velocity of quantum walks in rational electric fields which can be thought of a single-particle equivalent of the Lieb-Robinson velocity \cite{extail,sigalPropagationInformationQuantum2024}. We prove that this \emph{maximal (group) velocity} of split-step walks in rational electric fields decays exponentially in the period $m$ of the field. More precisely, we show that the maximal velocity is exactly $|a|^m$, where $a$ with $|a|<1$ is the transmission coefficient characterizing the walk. This bound extends the previous results from \cite{ARS2023-CMP} significantly, where the proof techniques limited the validity to an upper bound for transmission coefficients $|a|<1/4$ only.
We use a conjunction of the Fourier techniques from \cite{ewalks} and a generalized ``sieving'' technique which states that the product of any two shift-coin walks amounts to a direct sum of split-step walks.

Our results complement recent works on discrete Schrödinger operators with periodic potentials \cite{ADFS2024, ARFFW}. The velocity bounds established there demonstrate that the maximal velocity, while positive, can be made arbitrarily small by controlling the potential. Our findings similarly highlight that careful periodic modulation can systematically control and suppress transport velocities, offering significant insights into engineering quantum systems with tailored dynamical properties.

The paper is organized as follows. Section \ref{sec:Model-results} introduces quantum walks in the presence of electric fields and presents our main results on maximal velocity, revival relations, and spectral properties (Theorem \ref{thm:speed_limit_W}). These results are proved in Section \ref{sec:proof:main} using a Fourier transform and the sieving method developed in Theorem \ref{thm:sieving}. Section \ref{sec:velocity-sum} collects crucial intermediate results on velocity bounds for powers and direct sums of walks, that are used in the proof of Theorem \ref{thm:speed_limit_W}. Finally, the conjecture in \cite{ARS2023-CMP} that we prove is written in the setting of CMV matrices. To make explicit the connection between our results and those of \cite{ARS2023-CMP}, we discuss the connection between quantum walks and generalized extended CMV matrices in the Appendix.

\section{Model and results}\label{sec:Model-results}
\subsection{Quantum walks}\label{sec:QWs}

We consider quantum walks on the one-dimensional lattice with a two-dimensional internal degree of freedom. These systems are modeled as unitary operators on $\CH=\ell^2(\bbZ)\otimes\bbC^2$, where the two factors represent the position space and the two internal degrees of freedom of the quantum particle, respectively.
We label the basis of $\CH$ as $\delta_{n}^{s}=\delta_{n}\otimes e_{s}$ for $n\in\bbZ$ and $s\in\{+,-\}$
with the ordering 
\begin{equation}\label{def:ordering}
	\dots,\delta_{n-1}^{-},\delta_{n}^{+},\delta_{n}^{-},\delta_{n+1}^{+},\delta_{n+1}^{-},\dots
\end{equation}	
where $\{\delta_n:n\in\bbZ\}$ and $\{e_+=[1,0]^\top,e_-=[0,1]^\top\}$ are the standard bases of $\ell^2(\bbZ)$ and $\bbC^2$. 
On $\CH$, we define two types of unitary operators:
\begin{itemize}
	\item[(a)] \emph{state-dependent shift operators} $S_\pm$ which are defined via the bilateral shift $T:\delta_n\mapsto\delta_{n+1}$ on $\ell^2(\bbZ)$ and the orthogonal projections $P_{\pm}=|e_\pm\rangle\langle e_\pm|$ on $\bbC^2$ as
		\begin{equation}\label{eq:shift_pm}
			S_\pm=T^{\pm1}\otimes P_\pm+\idty_{\bbZ}\otimes P_\mp.	
		\end{equation}
	\item[(b)] \emph{coin operators} $C$ on $\CH$ which act locally at position $n$ via a unitary $2\times2$ matrix
		\begin{equation}\label{eq:coin_loc}
				C(n)=\begin{bmatrix}a(n)&b(n)\\c(n)&d(n)\end{bmatrix},	
		\end{equation}
		where the unitarity of $C(n)$ implies that $|a(n)|=|d(n)|$ and $|b(n)|=|c(n)|$ for all $n\in\bbZ$.
\end{itemize}
\smallskip

A \emph{quantum walk} on $\CH$ with strictly finite jump width is a finite product of such shifts and coins. This construction is exhaustive: any banded unitary operator can be decomposed into such a finite product, both in the general case \cite{shiftcoin} as well as in the translation-invariant setting \cite{vogtsDiscreteTimeQuantum2009}.
	
The concrete systems we study are \emph{split-step} quantum walks that are defined via
\begin{equation}\label{eq:SSwalk}
	W=S_+C_1S_-C_2,
\end{equation}
where the coins $C_j$ act locally as $C_j(n)$.
A special type of split-step walks is obtained by setting the first coin operator to the identity operation, i.e. $C_1=\idty$. The resulting \emph{shift-coin} walk has the form
\begin{equation}\label{eq:SCwalk}
	U=SC,
\end{equation}
where $S=S_+S_-$ and $C= C_2$. 
Throughout this manuscript, we consistently denote split-step walks by $W$ and reserve $U$ for shift-coin walks\footnote{The motivation for this notation is the following: In Theorem \ref{thm:sieving} below we will see that two time steps with a shift-coin walk (``double-$U$'') are related to a single step with a split-step walk ($W$).}. When the distinction is not essential, we shall denote general quantum walks on $\mathcal{H}$ simply by $W$.

We say that a walk $W$ or a coin $C$ is \emph{translation-invariant} if $W$ or $C$ commutes with lattice translations. In this setting, the local coins are independent of position, i.e., $C(n)\equiv C(0)$ for all $n\in\bbZ$. Without loss of generality, we omit a global phase such that $\det(C(n)) = 1$, which leads to
\begin{equation}\label{def:C}
	C(n) = \begin{bmatrix}
		a & b \\ -\bar{b} & \bar{a}
	\end{bmatrix},\qquad a, b \in \bbC;\ |a|^2 + |b|^2 = 1.
\end{equation}
A popular choice is the \emph{Hadamard coin} $C_H$ with $a=b=1/\sqrt2$.

In the translation-invariant setting, the dynamics is (partially) diagonalized by the Fourier transform 
\begin{equation}
	(\mathscr{F}\psi)(\theta) = \frac{1}{\sqrt{2\pi}} \sum_{n \in \bbZ} e^{-i\theta n} \psi(n),
\end{equation}
where $\psi(n)$ is an element of the local Hilbert space $\bbC^2$ at $n\in\bbZ$. It is straightforward to check that a translation-invariant walk $W$ becomes a matrix-valued multiplication operator in Fourier space, that is,
\begin{equation}
	(\mathscr{F} W \mathscr{F}^{-1}\hat\psi)(\theta) = \hat{W}(\theta)\hat\psi(\theta).
\end{equation}
Here, $\hat\psi(\theta)$ is $\bbC^2$-valued and $\hat W(\theta)$ is a unitary $2\times2$ matrix for every $\theta\in\bbT:=[0,2\pi)$. 
We can further diagonalize the unitary matrices $\hat W(\theta)$ pointwise which yields
\begin{equation}\label{eq:eigendecomposition}
	\hat W(\theta)=\sum_{s=\pm}e^{i\omega_s(\theta)}P_s(\theta),
\end{equation}
where the \emph{dispersion relations} $\omega_s(\theta)$ determine the eigenvalues and the $P_s(\theta)$ are (rank-1) orthogonal eigenprojections.
Below, the dispersion relations play a central role in determining the velocity of translation-invariant walks.

\subsection{Electric fields}
We introduce electric fields into quantum walks via the discrete minimal coupling scheme described in \cite{CGWW2019JMP}, where electric fields appear as commutation phases between discrete-time shifts and lattice translations. In the present setting on $\bbZ$, these commutation phases can be implemented as multiplication by a position-dependent phase \cite{CGWW2019JMP,ewalks}.
We only consider homogeneous and static electric fields $\Phi\in\bbT:=[0,2\pi)$ which are implemented by the diagonal unitary operators
\begin{equation}\label{field}
	F_\Phi=e^{i\Phi Q}, \qquad\tilde{F}_\Phi=\left(\idty_\bbZ\otimes \begin{bmatrix}1 & 0\\ 0 & e^{i\Phi}\end{bmatrix}\right) F_{2\Phi}
	=e^{i\frac{\Phi}{2}} e^{-i\frac{\Phi}{2}(\idty_\bbZ\otimes \sigma_3)} F_{2\Phi},
\end{equation}	
where $Q \delta_n^{\pm}=n\delta_n^{\pm}$ is the position operator on $\mathcal{H}$, and $\sigma_3$ is the third Pauli matrix. Note that the global phase $e^{i\Phi/2}$ can be dropped without loss of generality.

We mainly consider ``rational electric fields'' $\Phi=2\pi n/m$ with $n,m\in\bbZ$ coprime such that $e^{i\Phi}$ is a primitive $m^{\text{th}}$ root of unity. 
Since $\Phi/2\pi=n/m$ is rational, we refer to $F_{\Phi}$ and $\tilde{F}_\Phi$ as \emph{rational electric fields}. 
Throughout the paper, we mostly refer simply to ``the electric field $\Phi$'' whenever is clear from the context whether $F_\Phi$ or $\tilde{F}_\Phi$ is intended. 

We incorporate these rational electric fields into shift-coin and split-step quantum walk dynamics by defining
\begin{equation}\label{eq:electrify}
U_{\Phi}:=F_\Phi U=F_\Phi S C \quad\text{ and }\quad W_\Phi:=\tilde{F}_{\Phi} W=\tilde{F}_{\Phi} S_+ C S_- C,
\end{equation}
respectively. Dynamical and spectral properties of the electric shift-coin $U_\Phi$ have been explored previously in \cite{CGWW2019JMP, ewalks, locQuasiPer}. The crucial observation there was that, while $U_\Phi$ itself is not, the temporally regrouped walk $U_\Phi^m$ is translation-invariant.
This periodic structure of $U_\Phi^m$ leads to revivals in the dynamics under $U_\Phi$: in \cite[Theorem 1]{ewalks} it is established that
\begin{equation}\label{U-revivals}
	\left\|U_\Phi^{2m} + \idty\right\| = 2|a|^m \quad \text{for $m$ odd,}\qquad
	\left\|U_\Phi^{m} + (-1)^{m/2}\idty\right\| = 2|a|^{m/2} \quad \text{for $m$ even}.
\end{equation}
This ``revival theorem'' implies that $U_\Phi$ reproduces any initial state periodically up to an error that is exponentially small in $m$ such that the walker exhibits very little net transport.

The observation that the temporally regrouped walk commutes with the lattice translations opens up the toolbox of Fourier techniques. These tools are standard in the study of quantum walks (see, e.g.,~\cite{ambainisOnedimensionalQuantumWalks2001,grimmettWeakLimitsQuantum2004,ahlbrechtAsymptoticEvolutionQuantum2011,ewalks}). The proof of the revival theorem hinges on the observation that the dispersion relations of $U_\Phi^m$ can be explicitly calculated \cite{ewalks} and are given as $\omega_\pm(\theta,m)$ where
\begin{equation}\label{def:omega_intro}
	\cos(\omega_\pm(\theta,m)) =
	\begin{cases}
		|a|^m\cos(m(\theta+\arg(a))), & m \text{ odd}, \\[3pt]
		-|a|^m\cos(m(\theta+\arg(a))) + (-1)^{m/2}(|a|^m-1), & m \text{ even}.
	\end{cases}
\end{equation}
Note that in the special case here $\omega_-=-\omega_+$ because $\hat W(\theta) \in\SU(2)$ by \eqref{def:C}.
This implies immediately that $\text{spec}(U_\Phi)$ is absolutely continuous and consists of $2m$ bands. The dispersion relations depend explicitly on the coin parameter $a$; we suppress this dependence to simplify notation.

On the other hand, the electric split-step walk $W_\Phi$ with translation-invariant coins is precisely the model studied in~\cite{ARS2023-CMP}. Note that $\tilde F_\Phi$ is not precisely an electric field as defined in \cite{CGWW2019JMP} due to the additional $\exp[-i\Phi(\idty_\bbZ\otimes \sigma_3)/2]$. However, this translation-invariant coin can be absorbed and $W_\Phi$ can be written (up to a global phase $\exp[i\Phi/2]$) as
\begin{equation*}
	W_\Phi = F_{2\Phi} S_+ \tilde{C} S_- C, \quad \text{where} \quad \tilde{C} := e^{-i\frac{\Phi}{2}(\idty_\bbZ\otimes \sigma_3)} C.
\end{equation*}
This justifies calling $W_\Phi$ an electric walk. Note, that $2\Phi/(2\pi)=2m/n$ is reduced only if $n$ is odd.

\subsection{The speed limit of electric quantum walks}

A fundamental question for quantum walks is how fast information can propagate under the given dynamics, especially when external fields are present.
Quantum walks model the discrete-time dynamics of single particles. Thus, in $t \in \mathbb{N}$ timesteps, an initial state $\psi$ evolves under an arbitrary quantum walk $W$ as $\psi\mapsto W^t\psi$.
Moreover, let us denote by $Q(t):=W^{-t}QW^t$ the position operator at time $t$ in the Heisenberg picture. 
The velocity of a (normalized) state  $\psi\in\mathcal H$ with respect to $W$ is given as
\begin{equation}\label{def:v-psi0}
v(W,\psi):= \limsup_{t \to \infty} \frac{1}{t} \| Q(t)\psi\|\equiv\limsup_{t \to \infty} \frac{1}{t} \| QW^{t}\psi\|.
\end{equation}
Then we define the \emph{maximal velocity} as the supremum over all states $\psi \in \dom(Q)$ of the asymptotic expected position in ballistic scaling (recall that $Q$ is the position operator on $\mathcal H$):
\begin{equation}\label{def:v}
v(W) := \sup_{\substack{\psi \in \dom(Q)\\ \|\psi\|=1}} v(W,\psi).
\end{equation}
Clearly, for any $t \in \mathbb{N}$ the domain of $Q W^t$ contains all finitely supported states. Moreover, it is straightforward to check that $\dom(Q)$ is invariant\footnote{This follows from checking that if $\psi\in\dom(Q)$ then $C_j \psi$ and $S_\pm \psi$ are in $\dom(Q)$.} under $W$, that is, if $\psi\in\dom(Q)$ then $W^t \psi\in\dom(Q)$ for any $t\in\mathbb{N}$.

In \cite[Theorem 3.2]{ARS2023-CMP} it is proved that the velocity of any initial state $\Psi=\delta_0\otimes\psi(0)$ localized at $0$ with $\psi(0)\in\bbC^2$ normalized is bounded from above by 
\begin{equation}\label{eq:bound_HG}
v(W_\Phi,\Psi)\leq (4|a|)^m.
\end{equation}
This means that the velocity of $\Psi$ decays exponentially in $m$ whenever $|a|\in(0,1/4)$ whereas the bound becomes trivial for all other coin parameters.
Numerics led the authors of \cite{ARS2023-CMP} to conjecture that the exponential decay of the velocity holds for all $|a|\in[0,1]$. Part (a) of our main result implies this conjecture: it provides an exact bound for the velocity of \emph{any} initial state in the domain of $Q$ and, moreover, it improves the prefactor in \eqref{eq:bound_HG}:
\begin{theorem}\label{thm:speed_limit_W}
Let $m,n\in\mathbb{N}$ be coprime and let $\Phi=2\pi n/m$. Consider the electric split-step walk $W_\Phi=\tilde{F}_{\Phi} S_+ C S_- C$ with translation-invariant coin $C$. Then, for $|a|\in[0,1]$ as in \eqref{def:C}:
\begin{enumerate}[label=\textrm{(\alph*)}]
\item The maximal velocity of  $W_\Phi$  is	
		\begin{equation}\label{eq:v:SSWalk}
			v(W_{\Phi})=|a|^m.
		\end{equation}
\item $W_\Phi$ satisfies the revival relation
		\begin{equation}\label{W-revivals}
		\left\| W_\Phi^m+(-1)^{m-n}\idty\right\|=2|a|^m.
		\end{equation}
\item The spectrum of $W_\Phi$ is absolutely continuous and given as
\begin{equation}\label{eq:spec_W}
\mathrm{spec}(W_\Phi)=
\bigcup_{\substack{\theta \in \bbT,\:s=\pm\\k=0,1,\ldots,2m-1}}\left\{e^{i \frac{1}{m}(\widehat\omega_s(\theta,m)+\pi k)} \right\} \text{ where }
\widehat\omega_s(\theta,m)=
\begin{cases}
\omega_s(\theta,2m), & $n$ \text{ odd},\\
2\omega_s(\theta,m), &  $n$ \text{ even},
\end{cases}
\end{equation}
and $\omega_\pm(\theta,m)$ is given by \eqref{def:omega_intro}.
\end{enumerate}
\end{theorem}
Thus, the walker's ballistic spread becomes exponentially suppressed for all $|a| \in (0,1)$. Note that $v$ is an asymptotic quantity that does not establish a bound for finite times, see Figure \ref{fig:revivals}. However, it is known that the faster-than-light contributions, so the contributions with velocities larger than $v$, decay exponentially in time \cite{extail}.
\begin{remark}
	An upper bound on $v(W)$ automatically implies an upper bound on the asymptotic standard deviation of the time-evolved position operator in ballistic scaling: 
	\begin{equation}
		\sigma_{\psi}(Q(t)):=\sqrt{\langle \psi, Q(t)^2 \psi \rangle - \langle \psi, Q(t) \psi \rangle^2 } \leq \sqrt{\langle \psi, Q(t)^2 \psi \rangle } = \| Q(t) \psi \|.
	\end{equation}
	Thus, bounding the maximal velocity also provides control over the rate of quantum spreading, or uncertainty, in the system. See also Figure \ref{fig:revivals}.
\end{remark}

Theorem \ref{thm:speed_limit_W} is proved in Section \ref{sec:proof:main}.
While $\tilde{F}_\Phi (T^{\pm 1}\otimes \idty)=e^{2i\Phi} (T^{\pm 1}\otimes \idty) \tilde{F}_\Phi$ and thus $W_\Phi^{2m}$ is translation-invariant, a direct approach via Fourier technics as for the electric shift coin walk $U_\Phi$ in \cite{ewalks} does not seem feasible: there, the structure of the multiplication operator $U_\Phi^{2m}$ resp. $U_\Phi^m$ in Fourier space played a significant role in determining the exact form of the dispersion relations \eqref{def:omega_intro}. The Floquet matrix of $W_\Phi^{2m}$ does not posses this structure, which seems to render its dispersion relations inaccessible. For the proof of Theorem \ref{thm:speed_limit_W} we therefore embark a different, less direct approach: the key observation that $U_{\Phi}^2$ acts as a two-step walk suggests decomposing the position space into even and odd lattice sites. Then $U_{\Phi}^2$ decomposes as a direct sum of two split-step walks which allows us to establish a relation between the velocities of $U_\Phi$ and $W_\Phi$.

\begin{figure}[h]
	\begin{center}
		\includegraphics[width=.9\textwidth]{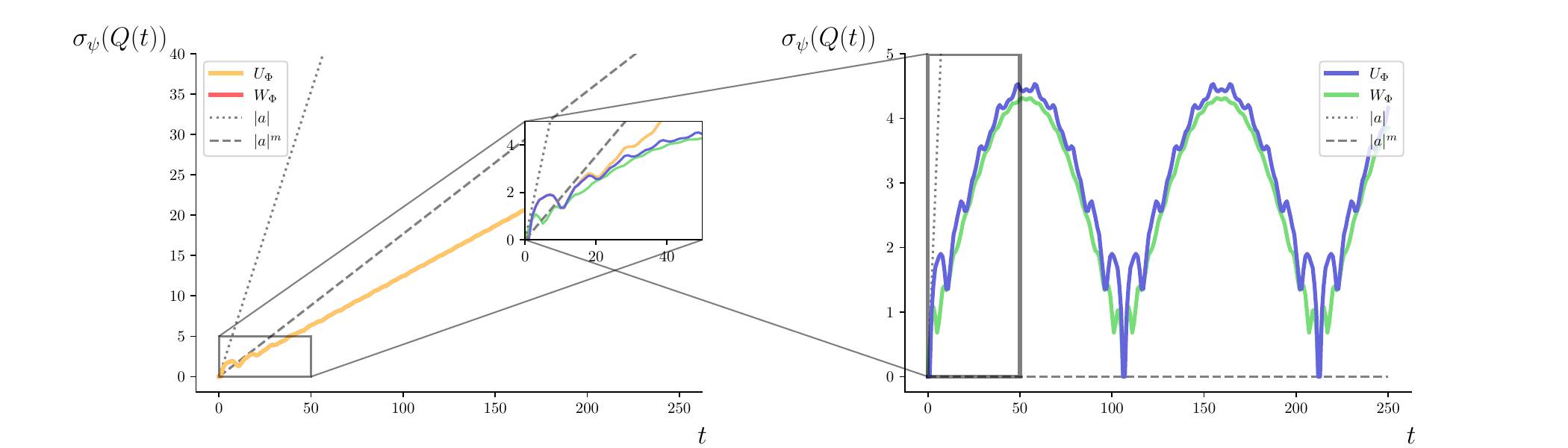}
	\end{center}
	\caption{\label{fig:revivals}The standard deviation of the dynamics under $U_\Phi$ and $W_\Phi$ with the Hadamard coin $C=C_H$ for $\Phi/(2\pi)=1/5$ (left) and $\Phi/(2\pi)=21/106=[0,5,21]$ (right) with initial state $\psi=|0\rangle\otimes[1,i]^\top/\sqrt{2}$. As the inset shows, the dynamics is initially very similar until about the order of $t=20$ where the errors committed in each revival accumulated enough such that the next term in the continued fraction expansion kicks in, see \cite{ewalks} for a thorough explanation of the interplay between the continued fraction expansion of $\Phi$ and the revivals.}
\end{figure}

The revival formula~\eqref{W-revivals} demonstrates that the system evolving under $W_\Phi$ returns close to its initial state, after every $m$ steps up to an error decaying exponentially in $m$. On the one hand, this mirrors the revivals in the shift-coin case. However, for $m$ odd $W_\Phi$ exhibits revivals after $m$ time steps, whereas the revivals for $U_\Phi$ occur only after $2m$ steps, see Figure \ref{fig:revivals}. These revivals underpin the exponential suppression of the maximal velocity and highlight the fundamental role of periodicity in suppressing ballistic transport.

As a direct consequence of the revival relation we conclude that irrational fields that are well-approximable in terms of their continued fraction expansion lead to hierarchical motion: an alternation between farther and farther excursion and better and better revivals on the time scales of the denominators of their continued fraction approximants. As argued in \cite{ewalks}, the excursions preclude point spectrum while the revivals preclude absolutely continuous spectrum via the Riemann-Lebesgue lemma, wherefore the spectral type of such walks must be singular continuous. On an arithmetic level, one can show this for example for Lebesgue fields via a Gordon-type argument, see \cite{CedzichPhD}. In the typical case, that is, for a full measure set of fields, Anderson localization follows for $W_\Phi$ from the proof in \cite{locQuasiPer} via sieving.

\begin{figure}[h]
	\def\a{1/sqrt(2)}
	\begin{center}
		\includegraphics[width=.9\textwidth]{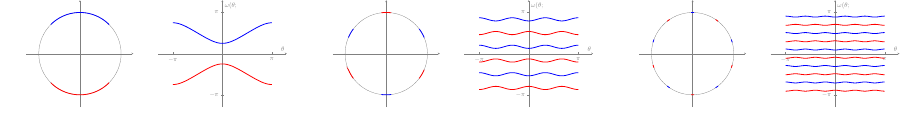}	
	\caption{\label{U-Phi-spectrum-m-odd}The spectra resp. the dispersion relations $\omega_\pm(\theta,m)$ of $U_\Phi$ for $m=1,3,5$ (left to right). The colors distinguish the different signs in \eqref{def:omega_intro}.}
	\end{center}
\end{figure}

\begin{figure}[h]
	\def\a{1/sqrt(2)}
	\begin{center}
		\includegraphics[width=.9\textwidth]{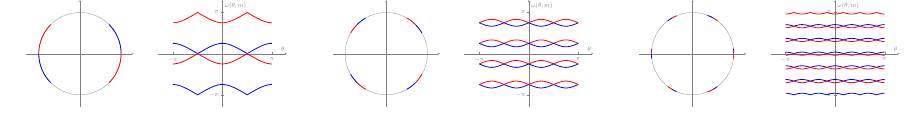}
	\end{center}
	\caption{\label{U-Phi-spectrum-m/2-odd}The spectra resp. the dispersion relations $\omega_\pm(\theta,m)$ of $U_\Phi$ for $m=2,4,6$ (left to right). The colors distinguish the different signs in \eqref{def:omega_intro}. Note that opposed to the case of $m$ odd, here the choices for $\omega_\pm$ are not analytic at the degenerate points. However, such an analytic choice is always possible in the present setting \cite{ahlbrechtAsymptoticEvolutionQuantum2011,ahlbrechtAsymptoticBehaviorDecoherent2013}.}
\end{figure}

\section{Proof of Theorem \ref{thm:speed_limit_W}}\label{sec:proof:main}
In this section, we prove the main result, Theorem \ref{thm:speed_limit_W}, by exploiting the decomposition of the Hilbert space $\mathcal{H}=\ell^2(\bbZ)\otimes\bbC^2$ into two subspaces corresponding to even and odd lattice sites. This decomposition enables us to relate the dynamics of split-step and shift-coin walk: the product of two shift-coin operators decomposes as a direct sum of two split-step walks with respect to this even-odd decomposition.
Combining this with a bound on the maximal velocity of electric shift-coin walks that follows from their dispersion relations in \eqref{def:omega_intro} is the key to establishing the exact form of velocity bounds and revival phenomena for split-step walks in rational electric fields.

\subsection{The maximal velocity of $U_\Phi$}

In the setting where $W$ is an arbitrary translation-invariant quantum walk, it is well known (see, e.g., \cite[Theorem 4]{ahlbrechtAsymptoticEvolutionQuantum2011}, \cite[Theorem 9.4]{damanikSpreadingEstimatesQuantum2016} or \cite[Proposition 2.2]{extail}) that for any $\psi \in \dom(Q)$ there exists a bounded self-adjoint operator $V$ on $\mathcal{H}$ such that
\begin{equation}\label{eq:lim}
	\lim_{t\rightarrow\infty} \frac{1}{t}Q(t)\psi= V\psi.
\end{equation}
This \emph{(group) velocity operator} $V$ is explicitly given as 
\begin{equation}
	V:=\mathscr{F}^{-1}\left[\sum_s\partial_\theta\omega_s(\theta)P_s(\theta)\right]\mathscr{F},
\end{equation}
where $\omega_s$ and $P_s(\theta)$ are the dispersion relations and the eigenprojections from \eqref{eq:eigendecomposition}, and the eigenvalues $\partial_\theta\omega_s(\theta)$ in Fourier space are called ``group velocities'' \cite{extail}.
Equation~\eqref{eq:lim} shows that the $\limsup$ in the definition of the velocity \eqref{def:v} is in fact just a limit. Thus, the velocity of $W$ can be expressed as the maximum of the group velocities:
\begin{align}\label{eq:v:bound:w}
	v(W) &= \sup_{\substack{\psi\in \dom(Q)\\ \|\psi\|=1}} \lim_{t\to\infty}\frac{1}{t}\|W^{-t} Q W^t \psi\| =\sup_{\substack{\psi\in \mathcal{H}\\ \|\psi\|=1}}\|V\psi\|=\max_{s=\pm,\ \theta\in\bbT} |\partial_\theta\omega_s(\theta)|,
\end{align}
where we used the fact that $\dom(Q)$ is dense in $\mathcal{H}$.

As mentioned above, the electric shift-coin walk $U_\Phi$ with rational field $\Phi=2\pi m/n$ is itself not translation-invariant, so to apply the above we need to consider the velocity of the $m$-fold regrouped walk $v(U_\Phi^m)$. For this walk, the dispersion relations are given in \eqref{def:omega_intro}, and its speed is obtained by computing:
\begin{lemma}\label{lem:d-omega-bound}
	For the coin matrix $C$ in \eqref{def:C}, we have that
	\begin{equation}
		\max_{\theta\in\bbT}|\partial_\theta\omega_\pm(\theta,m)| = \begin{cases}
			m |a|^m, & m \text{ odd}, \\[3pt]
			m |a|^{m/2},  &  m \text{ even},
		\end{cases}
	\end{equation}
	with the dispersion relations $\omega_\pm(\theta,m)$ given in \eqref{def:omega_intro}.
\end{lemma}
\begin{proof}[Proof of Lemma~\ref{lem:d-omega-bound}]
	Differentiating the dispersion relations from \eqref{def:omega_intro} with respect to $\theta$ and solving for $\partial_\theta\omega_\pm(\theta,m)$ gives
	\begin{equation}
		\partial_\theta\omega_\pm(\theta,m) = (-1)^{m+1} \frac{m|a|^m \sin(m(\theta+\arg(a)))}{\sin(\omega_\pm(\theta,m))}.
	\end{equation}
	We square this expression and use \eqref{def:omega_intro} again. Considering the odd and even cases for $m$ separately yields
	\begin{description}
		\item[$m$ odd] Set $y_1(\theta):=\cos^2(m(\theta+\arg(a)))\in[0,1]$ such that
		\begin{equation}
			(\partial_\theta\omega_\pm(\theta,m))^2 = m^2 |a|^{2m} \frac{1-\cos^2(m(\theta+\arg(a)))}{1 - |a|^{2m}\cos^2(m(\theta+\arg(a)))}=m^2 |a|^{2m} \dfrac{1-y_1(\theta)}{1-|a|^{2m}y_1(\theta)}.
		\end{equation}
		\item[$m$ even] Set $y_2(\theta) := 1 - (-1)^{m/2} \cos(m(\theta+\arg(a)))\in[0,2]$. Then from~\eqref{def:omega_intro}
		\begin{equation}
			\cos(\omega_\pm(\theta,m)) = (-1)^{m/2+1}\left( 1 - |a|^m y_2(\theta) \right).
		\end{equation}
		Thus,
		\begin{equation}
			\sin^2(\omega_\pm(\theta,m)) = |a|^m y_2(\theta)\left(2 - |a|^m y_2(\theta)\right),
		\end{equation}
		and hence,
		\begin{equation}
			(\partial_\theta\omega_\pm(\theta,m))^2 = m^2 |a|^{m} \frac{2 - y_2(\theta)}{2 - |a|^m y_2(\theta)}.
		\end{equation}
	\end{description}
	The functions
	\begin{equation}
		f_1(y) = \frac{1-y}{1-|a|^{2m}y}, \qquad f_2(y) = \frac{2-y}{2-|a|^{m}y}
	\end{equation}
	are decreasing for all $y \in [0,2]$ and hence achieve their maximum value $1$ at $y=0$. 
\end{proof}
Together with Lemma \ref{lem:speed_power} below, this implies that the velocity of $U_\Phi$ is given by
\begin{equation}\label{U-v-bound}
	v(U_\Phi) =\frac1mv(U_\Phi^m)=
	\begin{cases}
		|a|^{m} & m \text{ odd},\\
		|a|^{m/2} & m \text{ even},
	\end{cases}
\end{equation}
where $|a|\in[0,1]$ is the coin parameter from~\eqref{def:C}.

\subsection{Sieving split-step walks}

Our methods to prove the main result in Theorem  \ref{thm:speed_limit_W} depend on the decomposition of the lattice over which $\CH$ is defined into even and odd labeled cells as
\begin{equation}\label{dec:H}
	\CH=\CH_e\oplus\CH_o:=\Big[\ell^2(2\bbZ)\otimes\bbC^2\Big]\oplus\Big[\ell^2(2\bbZ+1)\otimes\bbC^2\Big],
\end{equation}
where we choose bases $\{\delta^\pm_{e,n}\}$ and $\{\delta^\pm_{o,n}\}$ for $\CH_e$ and $\CH_o$, respectively, with ordering
\begin{equation}\label{eq:basis_W12_W12}
	(\delta^+_{e,n},\delta^-_{e,n}):=(\delta^+_{2n},\delta^-_{2n})\qquad
	(\delta^+_{o,n},\delta^-_{o,n}):=(\delta^+_{2n+1},\delta^-_{2n+1}).
\end{equation}		
\begin{remark}
For simplicity, we use the same symbols ($Q$, $W$, $S_\pm$, $C$, $\ldots$) to denote operators defined both on the full Hilbert space $\mathcal{H}$ and on its subspaces  $\mathcal{H}_e$ and $\mathcal{H}_o$. In each case, the symbol refers to the corresponding operator acting on the relevant space. This slight abuse of notation should cause no confusion, as the relevant domains will always be clear from the context.
\end{remark}

The following theorem shows that the square of a shift-coin walk decomposes as a direct sum of two split-step walks with respect to the even-odd decomposition in \eqref{dec:H}. This result, which we prove for completeness in Section \ref{sec:proof-sieving} below, is the translation of folklore knowledge in the literature of CMV matrices, see for example \cite{fillmanSpectralApproximationErgodic2017}.
\begin{theorem}\label{thm:sieving}
		Let $U=SC$, be a shift-coin walk on $\mathcal H=\ell^2(\bbZ)\otimes\bbC^2$. Then
		\begin{equation}\label{eq:U2=W+W}
			U^2=W\oplus \tilde W,
		\end{equation}
		where $W=S_+C_1S_-C_2$ and $\tilde W=S_+ \tilde C_1 S_- \tilde C_2$ are specified by the local coins
		\begin{equation*}
			 C_1(n)=C(2n+1),\quad C_2(n)=C(2n),\quad
			 \tilde C_1(n)=C(2n+2),\quad \tilde C_2(n)=C(2n+1),
		\end{equation*}		
		and the direct sum is with respect to even and odd lattice sites \eqref{dec:H}.
\end{theorem}
Note that when $C$ is translation-invariant, \eqref{eq:U2=W+W} reads as $U^2=W\oplus W$ where $W=S_+ C S_- C$.

The introduction of electric fields depends on which quantum walk we regard as a fundamental shift in time direction \cite{CGWW2019JMP}. For example, it makes a structural difference whether we electrify $U$ or $U^2$ via \eqref{eq:electrify}:
Since $F_\Phi$ is block-diagonal with each block equal to a phase times the identity, it commutes with coins, whereas $SF_\Phi =(\idty_\bbZ\otimes e^{-i\Phi\sigma_3})F_\Phi S$, where $\sigma_3$ denotes the third Pauli matrix. This directly implies that for a shift-coin walk $U$
\begin{equation*}
	U F_\Phi =e^{-i\Phi(\idty_\bbZ\otimes \sigma_3)}F_\Phi U.
\end{equation*}
Thus, squaring an electric shift-coin walk $U_\Phi=F_\Phi U$ yields
\begin{equation}\label{eq:scsc_fields}
	U_{\Phi}^2=F_\Phi UF_\Phi U=e^{-i\Phi(\idty_\bbZ\otimes \sigma_3)} F_{2\Phi} U^2=e^{-i\Phi(\idty_\bbZ\otimes \sigma_3)} F_{2\Phi} \left(W\oplus \tilde W\right)
\end{equation}
where $W$ and $\tilde W$ are as in Theorem  \ref{thm:sieving}. To fully decompose $U_{\Phi}^2$ with respect to even and odd lattice sites, it remains to decompose the field $F_\Phi=e^{i\Phi Q}$.
Clearly, the position operator $Q$ on $\mathcal{H}$ decomposes as
\begin{equation}\label{dec:Q}
	Q=(2Q)\oplus (2Q+\idty).
\end{equation}
This implies that the field $F_\Phi$ in \eqref{field} decomposes with respect to \eqref{dec:H} as
\begin{equation}\label{dec:F}
F_\Phi=e^{i\Phi Q}=e^{2i\Phi Q}\oplus e^{i\Phi} e^{2i\Phi Q}=F_{2\Phi}\oplus e^{i\Phi} F_{2\Phi}.
\end{equation}
Substituting this expression into \eqref{eq:scsc_fields} and replacing $\Phi\mapsto\Phi/2$ yields the following corollary:
\begin{coro} \label{coro:sieving}
Let $U_{\Phi}=F_{\Phi} U$ be a shift-coin walk subject to the electric field $F_\Phi=e^{i\Phi Q}$. Then
\begin{equation}\label{eq:decomp_electric_walks}
U_{\Phi/2}^2=\left(e^{-i\Phi/2} \tilde{F}_{\Phi} W\right)\oplus\left(e^{i\Phi/2}\tilde{F}_{\Phi} \tilde W\right)
\end{equation}
where $W$ and $\tilde W$ are the split-step walks as in Theorem \ref{thm:sieving} and the field $\tilde{F}_{\Phi}$ is given in \eqref{field}.
\end{coro}

This decomposition allows us to relate the velocities of the two types of walks: Lemma \ref{lem:W-W1-W2} below shows that the velocity of any walk that acts independently on even and odd lattice sites without mixing them is twice the maximum of the subsystem velocities. Moreover, it is intuitively clear that the speed of some $k^{th}$ power of a quantum walk is $k$-times the speed of the walk, see Lemma \ref{lem:speed_power} below. 
Applying these results to \eqref{eq:decomp_electric_walks} in the translation-invariant case yields, on the one hand
\begin{equation}\label{U2-W}
v\left(U_{\Phi/2}^2\right)=2\max\left\{v\left(e^{-i\Phi/2} W_{\Phi}\right), v\left(e^{i\Phi/2} W_{\Phi}\right)\right\}=2v\left(W_{\Phi}\right),
\end{equation}
and on the other (see Lemma \ref{lem:speed_power})
\begin{equation}\label{U2-2U}
v\left(U_{\Phi/2}^2\right)=2v\left(U_{\Phi/2}\right),
\end{equation}
which together imply that the velocities of $U_{\Phi/2}$ and $W_\Phi$ agree. This chain of arguments highlights why in the definition of $W_\Phi$ we use the field $\tilde F_\Phi$ instead of $F_\Phi$: Lemma \ref{lem:speed_power} only applies to powers of operators. If we would consider the split-step walk with field $F_\Phi$ instead of $\tilde F_\Phi$, the left side of \eqref{eq:decomp_electric_walks} we amount to $e^{i\Phi/2(\idty\otimes\sigma_3)}U_{\Phi/2}^2$ for which we have no handle on its velocity.

To prove Theorem \ref{thm:speed_limit_W} we exploit \eqref{eq:decomp_electric_walks}, where we note that for the electric field $\Phi=2\pi n/m$, the fundamental period of $\Phi/2$ on the left side is $\ell:=2m/\gcd(n,2)$. Since $n,m$ are coprime by assumption, it is easy to see that
\begin{equation}\label{eq:m-n}
	\ell=\begin{cases} 2m, &  n\text{ odd},\\ m, & n\text{ even},\end{cases}
\end{equation}
where in the second case $m$ is automatically odd.
Then part (a) of Theorem \ref{thm:speed_limit_W}, that is $v(W_\Phi)= |a|^{m}$ follows directly from \eqref{U-v-bound}.

To prove the revival relations \eqref{W-revivals} as well as part (c) of Theorem \ref{thm:speed_limit_W} about the spectrum of $W_\Phi$, we start from \eqref{eq:decomp_electric_walks}. It directly follows that
\begin{equation}\label{eq:UWPhi-revivals}
U_{\Phi/2}^{2m}-(-1)^{m+n}\idty=\left(e^{-im\Phi/2} W_\Phi^m-(-1)^{m+n}\idty\right)\oplus\left(e^{im\Phi/2} W_\Phi^m-(-1)^{m+n}\idty\right).
\end{equation}
Observe that for $\Phi=2\pi n/m$ with $m,n$ coprime,  $e^{\pm i m\Phi/2}=(-1)^n$. Taking the norm of \eqref{eq:UWPhi-revivals} yields
\begin{equation}
\left\|U_{\Phi/2}^{2m}-(-1)^{m+n}\idty \right\|=\left\| W_\Phi^m-(-1)^{m}\idty \right\|.
\end{equation}
Again, since $\exp(i\Phi/2)$ is a primitive $\ell^{th}$-root of the unity, \eqref{eq:m-n} shows that the spectral properties of 
$W_\Phi$ depend on the parity of $n$.
\begin{description}
\item[$n$ odd] In this case, $\ell=2m$ is even. Applying the revival formula \eqref{U-revivals} yields
\begin{equation}\label{pf:revivals:n-odd}
\left\|W_\Phi^m-(-1)^{m}\idty \right\|=\left\|U_{\Phi/2}^{2m}+(-1)^m\idty \right\|=2|a|^m.
\end{equation}
Moreover, 
\begin{equation}
	\text{spec}(W_{\Phi}^{2m}) = \text{spec}\big((U_{\Phi/2}^{2m})^2\big) = \bigcup_{\theta \in \bbT} \left\{ e^{i 2\omega_\pm(\theta,2m)} \right\},
\end{equation}
where $\omega_\pm(\theta,m)$ is given in~\eqref{def:omega_intro}, and for $2m$ the formula corresponds to the even case.
\item[$n$ even] In this case, $m$ is odd. Thus, again by \eqref{U-revivals},
\begin{equation}\label{pf:revivals:n-even}
\|W_\Phi^m+\idty\|=\|U_{\Phi/2}^{2m}+\idty\|=2|a|^m.
\end{equation}
Moreover,
\begin{equation}
	\text{spec}(W_{\Phi}^{2m}) = \text{spec}\big((U_{\Phi/2}^{m})^4\big) = \bigcup_{\theta \in \bbT} \left\{ e^{i 4\omega_\pm(\theta,m)} \right\}.
\end{equation}
\end{description}
In both cases, \eqref{eq:spec_W} follows by taking the $2m^{\text{th}}$ root.
\hfill\qedsymbol

\subsection{Proof of Theorem \ref{thm:sieving}}\label{sec:proof-sieving}
We begin by explicitly describing the action of a split-step quantum walk on basis elements $\{\delta_n^s: n\in\bbZ, s\in\{+,-\}\}$ with the ordering \eqref{def:ordering}. This will allow us to easily identify even-odd decompositions.
\begin{lemma}\label{lem:W_basis_rep}
Let $C_1(n),C_2(n)$, $n\in\bbZ$ be two sequences of unitary $2\times2$ matrices that serve as local coins of the form \eqref{eq:coin_loc}. Then, the corresponding split-step walk $W$ acts on the basis $\{\delta_n^s:n\in\bbZ,s=\pm\}$ of $\ell^2(\bbZ)\otimes\bbC^2$ as
\begin{align}\label{eq:Wss_action}
W\delta_{n}^+&= a_2(n)\left[a_1(n)\delta_{n+1}^++c_1(n)\delta_n^-\right]+c_2(n)\left[b_1(n-1)\delta_{n}^++d_1(n-1)\delta_{n-1}^-\right] \notag\\
W\delta_{n}^-&= b_2(n)\left[a_1(n)\delta_{n+1}^++c_1(n)\delta_n^-\right]+d_2(n)\left[b_1(n-1)\delta_{n}^++d_1(n-1)\delta_{n-1}^-\right].
\end{align}
\end{lemma}
\begin{proof}
	This follows directly from the definition of the state-dependent shift operators in \eqref{eq:shift_pm} and the coins in \eqref{eq:coin_loc}.
\end{proof}
	
Let $U$ be a shift-coin walk with local coins as in \eqref{eq:coin_loc}. To prove Theorem \ref{thm:sieving}, we start by noting the invariance of the even and odd subspaces $\mathcal{H}_e$ and $\mathcal{H}_o$ under the action of $U^2$. Precisely, 
\begin{equation}
U^2 \mathcal{H}_e\subset\mathcal{H}_e\quad\text{ and }\quad U^2 \mathcal{H}_o\subset\mathcal{H}_o.
\end{equation}	
This can be seen from calculating the action on the basis states: at the even lattice sites the product acts as
\begin{align*}
U^2\delta_{2n}^+ & = a(2n)[a(2n+1)\delta_{2n+2}^++c(2n+1)\delta_{2n}^-]+ c(2n)[b(2n-1)\delta_{2n}^++d(2n-1)\delta_{2n-2}^-] \notag\\
U^2\delta_{2n}^- &= b(2n)[a(2n+1)\delta_{2n+2}^++c(2n+1)\delta_{2n}^-]+d(2n)[b(2n-1)\delta_{2n}^++d(2n-1)\delta_{2n-2}^-] 
\end{align*}
Similarly, on the odd lattice sites,
\begin{align*}
U^2\delta_{2n+1}^+ &= a(2n+1)[a(2n+2)\delta_{2n+3}^++c(2n+2)\delta_{2n+1}^-]+c(2n+1)[b(2n)\delta_{2n+1}^++d(2n)\delta_{2n-1}^-] \notag \\
U^2\delta_{2n+1}^-&= b(2n+1)[a(2n+2)\delta_{2n+3}^++c(2n+2)\delta_{2n+1}^-]+d(2n+1)[b(2n)\delta_{2n+1}^++d(2n)\delta_{2n-1}^-]
\end{align*}
Comparing with \eqref{eq:Wss_action}, one immediately sees that the action of $U^2$ on each subspace $\mathcal{H}_e$ and $\mathcal{H}_o$ coincides with that of two separate split-step quantum walks with local coin operators as described in the statement of Theorem \ref{thm:sieving}.
	
\begin{remark}
	The ordering of the basis in $\mathcal H_e$ and $\mathcal H_o$ is important here: taking instead
	\begin{equation}\label{eq:basis_W12_W21T}
	(\delta^+_{e,n},\delta^-_{e,n}):=(\delta^+_{2n},\delta^-_{2n})\qquad
	(\delta^+_{o,n},\delta^-_{o,n}):=(\delta^-_{2n-1},\delta^+_{2n+1}),
	\end{equation}
	and comparing with the analogue of Lemma \ref{lem:W_basis_rep} for the transposed split-step walk, we find that $U^2$ acts on $\mathcal H_o$ as $W^\top=C_2^\top S_-^\top C_1^\top S_+^\top$ with local coins
	\begin{equation}
		C_1(n)=\sigma_1C(2n+1)^\top\sigma_1,\qquad C_2(n)=\sigma_1C(2n)^\top\sigma_1.
	\end{equation}
\end{remark}

\subsection{Speed limit of direct sums and powers of  quantum walks}\label{sec:velocity-sum}
In this section, we establish two structural results on the propagation velocity of quantum walks that play a central role in the proof of Theorem \ref{thm:speed_limit_W} and are of independent interest. Both of them are expectable: Lemma \ref{lem:W-W1-W2} says that the velocity of a direct sum of quantum walks is determined by the fastest component, while Lemma \ref{lem:speed_power} establishes that the velocity of the $k^{\text{th}}$ power of a quantum walk is $k$-times the velocity of the walk. Together, these results are useful for analyzing composite or periodically driven quantum walks and underpin our analysis of ballistic propagation.

\begin{lemma}\label{lem:W-W1-W2}
	Let $W_1$ and $W_2$ be arbitrary quantum walks on $\ell^2(\bbZ)\otimes\bbC^2$. Then, the velocity of their direct sum $W_1\oplus W_2$ with respect to even and odd lattice sites \eqref{dec:H} equals twice the maximum of the velocities of $W_1$ and $W_2$, i.e.,
	\begin{equation}\label{eq:W-W1-W2}
		v\left(W_1\oplus W_2\right)=2\max\{v(W_1),v(W_2)\}.
	\end{equation}
\end{lemma}	

\begin{proof}
First recall the decomposition of the position operator from \eqref{dec:Q}. Writing  $W_{12}:=W_1\oplus W_2$ we observe that for any $t\in\mathbb{N}$ and normalized $\psi=\psi_e\oplus\psi_o$, we have
		\begin{equation}
			QW_{12}^t\psi=2QW_{1}^t\psi_e\oplus(2Q+\idty)W_{2}^t\psi_o,
		\end{equation}
The velocity of $\psi$ with respect to $W_{12}$ becomes
		\begin{align}\label{pf:v-1}
			v(W_{12},\psi)^2&=\limsup_{t\to\infty}\frac{1}{t^2}\Big[\left\|2QW_{1}^t\psi_e\right\|^2+\left\|(2Q+\idty)W_{2}^t\psi_o\right\|^2\Big].
		\end{align}
		We emphasize that the terms on the right hand side are not velocities in general since $\psi_e$ and $\psi_o$ might not be normalized.
		
For the second term on the right of \eqref{pf:v-1}. Two direct applications of triangle inequalities yield
\begin{equation}
\left\|2QW_2^t \psi_o\right\|-\|\psi_o\|\leq \left\|(2Q+\idty)W_{2}^t\psi_o\right\|\leq \left\|2QW_2^t \psi_o\right\|+\|\psi_o\|.
\end{equation}
Use this inequality in (\ref{pf:v-1}) together with the basic fact that (observe that $\frac{1}{t}\|\psi_o\|\to 0$)
\begin{equation}\label{eq:basic-fact}
	\lim_{t\rightarrow\infty} g(t)=0 \quad\Rightarrow\quad \limsup_{t\to\infty}(f(t)+g(t))=\limsup_{t\to\infty}f(t),
\end{equation}
to see that
\begin{equation}\label{pf:main}
v(W_{12},\psi)^2 =4\limsup_{t\to\infty}\frac{1}{t^2}\Big[\|QW_{1}^t\psi_e\|^2+\|Q W_{2}^t\psi_o\|^2\Big].
\end{equation}
In particular, when $\psi$ is supported only on $\CH_e$, i.e. $\psi=\psi_e\oplus 0_{\CH_o}$, and $\|\psi\|=\|\psi_e\|=1$, \eqref{pf:main} reduces to 
\begin{equation}\label{pf:psi-0-0}
2v(W_1,\psi_e)=v(W_{12},\psi_e\oplus 0_{\CH_o}).
\end{equation}
Take the supremum over all normalized initial states $\psi_e\in \dom(Q)$, to see that
\begin{equation}
2v(W_1)= \sup_{\psi_e} v(W_{12},\psi_e\oplus 0_{\CH_o})\leq \sup_{\psi}  v(W_{12},\psi)=v(W_{12})
\end{equation}
Similarly, we have $2v(W_2)\leq v(W_{12})$, and hence
\begin{equation}\label{pf:geq}
v(W_{12})\geq 2\max\{v(W_1),v(W_2)\}.
\end{equation}
To prove the other direction of the inequality \eqref{pf:geq}, we start with \eqref{pf:main} and we consider two cases for the state $\psi\in \dom(Q)$
\begin{itemize}
\item[case 1:] $\|\psi_\star\|\in\{0,1\}$, where $\star\in\{e,o\}$, i.e., $\psi=\psi_e\oplus 0_{\CH_o}$ or $\psi=0_{\CH_e}\oplus \psi_o$.
\item[case 2:] $\|\psi_\star\|\notin\{0,1\}$, where $\star\in\{e,o\}$, i.e.,  $\psi=\psi_e\oplus\psi_o\ \in \dom(Q)$ with nonzero $\psi_e$ and $\psi_o$.
\end{itemize}
For case 1, observe that \eqref{pf:psi-0-0} gives
\begin{equation}
v(W_{12},\psi_e\oplus 0_{\CH_o})=2v(W_1,\psi_e)\leq 2v(W_1).
\end{equation}
and similarly,
\begin{equation}
v(W_{12},0_{\CH_e}\oplus \psi_o)=2v(W_2,\psi_o)\leq 2v(W_2).
\end{equation}
For case 2:  we distribute the limsup in \eqref{pf:main} to obtain
\begin{align}
(v(W_{12},\psi))^2 &\leq 4\|\psi_e\|^2 \left(v(W_1,\frac{\psi_e}{\|\psi_e\|})\right)^2+4\|\psi_o\|^2 \left( v(W_2,\frac{\psi_o}{\|\psi_o\|})\right)^2\\
&\leq 4\max\left\{ (v(W_1))^2,(v(W_2))^2\right\}
\end{align}
where we used the fact that $\|\psi_e\|^2+\|\psi_o\|^2=1$. 

Combining the two cases, we find that for initial state $\psi\in \dom(Q)$,
\begin{equation}
v(W_{12},\psi)\leq 2\max\{v(W_1),v(W_2)\},
\end{equation}
which shows that
\begin{equation}\label{pf:leq}
v(W_{12})\leq 2\max\{v(W_1),v(W_2)\}.
\end{equation}
Combining \eqref{pf:geq} and \eqref{pf:leq} yields the desired equality \eqref{eq:W-W1-W2}.
\end{proof}

We note in passing that it is immediate from the proof that Lemma \ref{lem:W-W1-W2} carries over to more general quantum walks, for example, to those with internal degree of dimensions larger than $2$.


\begin{lemma}\label{lem:speed_power}
	Let $W$ be an arbitrary quantum walk. Then
	\begin{equation*}
		v(W^k)= k\,v(W).
	\end{equation*}
\end{lemma}

\begin{proof}
For a fixed $k\in\mathbb{N}$, the one direction that $v(W^k)\leq k\,v(W)$ follows from the following simple argument: 
For any normalized $\psi\in \dom(Q)$, recall that
\begin{equation}
v(W^k,\psi)=\limsup_{t\to\infty}\frac1t\|Q W^{kt}\psi\|,
\end{equation}
Since $(kt)_{t\in\mathbb N}$ is a subsequence of $(t)_{t\in\bbN}$, we obtain
\begin{equation}
v(W^k,\psi)\leq kv(W,\psi)\leq k v(W).
\end{equation}
Taking the $\sup$ shows that $v(W^k)\leq k\,v(W)$.

To prove the other direction, note that for any $t\in\mathbb{N}$, there are unique $s\in\mathbb{N}$ and $r\in\{0,1,\ldots,k-1\}$ such that $t=k s+r$. 
Then use Lemma \ref{lem:t-ks+r} below with $f(t):=t^{-1}\|Q W^t\psi\|$, $\psi\in\dom(Q)$ and $\|\psi\|=1$, to see that
\begin{align}\label{eq:W-ks+r:1}
v(W,\psi)&=\max_{0\leq r< k}\limsup_{s\rightarrow\infty}\frac{1}{ks+r}\left\|Q W^r W^{k s}\psi\right\| \notag\\
&\leq\frac{1}{k}\max_{0\leq r<k}\limsup_{s\rightarrow\infty}\frac{1}{s}\left(\left\|Q W^{ks}\psi\right\|+\left\|[Q,W^r]W^{ks}\psi\right\|\right).
\end{align}
Note that the second term in \eqref{eq:W-ks+r:1} is bounded as
\begin{equation}
\frac{1}{s}\left\|[Q,W^r]W^{ks}\psi\right\|\leq \frac{1}{s}\|[Q,W^r]\|\leq \frac{r}{s}\|[Q,W]\| \longrightarrow 0 \text{ as }s\to\infty
\end{equation}
The first-to-second step follows from expanding the commutator, and the last step follows from the boundedness of $[Q,W]$.
Thus, inequality \eqref{eq:W-ks+r:1} reduces to
\begin{equation}
v(W,\psi)\leq \frac{1}{k}\limsup_{s\rightarrow\infty}\frac{1}{s}\left\|Q W^{ks}\psi\right\|=\frac{1}{k}v(W^k,\psi)\leq \frac{1}{k}v(W^k).
\end{equation}
Taking the $\sup$ over normalized states in $\dom(Q)$ finishes the proof of the other direction 
$v(W^k)\geq k v(W)$.
\end{proof}

\begin{lemma}\label{lem:t-ks+r}
	Let $k\in\mathbb{N}$ be fixed and let $f:\mathbb{N}\to(\mathbb{R}\cup\{\pm\infty\})$ be any sequence. Then
	\begin{equation}\label{eq:t-ksr}
		\limsup_{t\to\infty}f(t)=\max_{0\leq r< k}\limsup_{s\to\infty}f(ks+r).
	\end{equation}
\end{lemma}
Note that this lemma intuitively follows from the fact that
\begin{equation}
	\mathbb{N}=\bigcup_{t\in\mathbb{N}}\{t\}=\bigcup_{0\leq r<k}\bigcup_{s\in\mathbb{N}}\{ks+r\}.
\end{equation}

\begin{proof}
	For every fixed  $r_0\in\{0,1,\ldots,k-1\}$, $(ks+r_0)_{s\in\mathbb{N}}$ is a subsequence of $(t)_{t\in\mathbb{N}}$, so
	\begin{equation}
		\limsup_{t\rightarrow\infty} f(t)\geq \limsup_{s\rightarrow\infty}f(ks +r_0)
	\end{equation}
	Since this holds for every $r_0\in\{0,1,\ldots,k-1\}$, then 
	\begin{equation}\label{eq:f:>}
		\limsup_{t\rightarrow\infty} f(t)\geq \max_{0\leq r<k} \limsup_{s\rightarrow\infty}f(ks +r).
	\end{equation}
	The other direction is more involved. Let
	\begin{equation}
		L:=\limsup_{t\rightarrow\infty} f(t)
	\end{equation}
	By definition of the $\limsup$, given any $\epsilon>0$, there are infinitely many $t\in\mathbb{N}$ such that $L-\epsilon<f(t)$. These values of $t$ define a subsequence $(t_n)_{n\in\mathbb{N}}$ of $(t)_{t\in\mathbb{N}}$ such that
	\begin{equation}\label{eq:L-xn}
		L-\epsilon< f(t_n)
	\end{equation}
	for all $n\in\mathbb{N}$. For each $t_n$, there is an $s_n\in\mathbb{N}$ and $r_n\in\{0,1,\ldots,k-1\}$, such that $t_n=k s_n+r_n$. Since there are only $k$ possible different values of $r_n$, then for any fixed $r_0\in\{0,1,\ldots,k-1\}$, there is a subsequence $(t_{n_m})$ of $(t_n)$ and a sequence $(s_{n_m})\in\mathbb{N}\rightarrow\infty$ of distinct values (it is a subsequence of $(s)_{s\in\mathbb{N}}$) such that
	\begin{equation}
		t_{n_m}=ks_{n_m}+r_0.
	\end{equation} 
	Hence \eqref{eq:L-xn} yields that
	\begin{align}\label{eq:L-snm}
		L-\epsilon &\leq \limsup_{m\rightarrow\infty} f(t_{n_m})=\limsup_{m\rightarrow\infty} f(k s_{n_m}+r_0) 
		\leq \limsup_{s\rightarrow\infty} f(k s+r_0) \notag\\
		&\leq \max_{0\leq r<k}\limsup_{s\rightarrow\infty} f(k s+r)
	\end{align}
	Since $\epsilon>0$ is arbitrary then 
	\begin{equation}\label{eq:f:<}
		\limsup_{t\rightarrow\infty} f(t)\leq \max_{0\leq r<k} \limsup_{s\rightarrow\infty}f(ks +r).
	\end{equation}
	\eqref{eq:f:>} and \eqref{eq:f:<} prove \eqref{eq:t-ksr}.
\end{proof}

\appendix

\section{Connection to generalized extended CMV matrices}
The the conjecture from \cite{ARS2023-CMP} we are proving and the bound \eqref{eq:bound_HG} that we are improving are formulated in the language of so-called \emph{Cantero-Moral-Velázquez} (CMV) matrices, which are objects of primary interest in the theory of orthogonal polynomials \cite{Simon2005OPUC1,Simon2005OPUC2}. 
There is a many-to-one correspondence between the objects of interest here, the split-step walks introduced in Section \ref{sec:QWs} and CMV matrices, called the \emph{CGMV connection} due to the foundational papers \cite{CGMV2012QIP,canteroMatrixvaluedSzegoPolynomials2010}. It was generalized to the setting of split-step quantum walks (resp. generalized extended CMV matrices) in \cite{CFO1,CFLOZ}.

For the convenience of the reader and to directly connect our results to those of \cite{ARS2023-CMP} we define these objects here and briefly comment on how their dynamical properties, particularly their maximal velocities and revival phenomena, can be inferred from those of quantum walks with electric fields. To this end, we call $(\alpha, \rho)\in \bbS^3 = \{(z_1, z_2) \in \bbC^2 : |z_1|^2 + |z_2|^2 = 1\}$ a \emph{Verblunsky pair} and for such a pair define the unitary matrix
\begin{equation}\label{def:Theta}
	\Theta(\alpha, \rho) := \begin{bmatrix} \overline{\alpha} & \rho \\ \overline{\rho} & -\alpha \end{bmatrix},
\end{equation}
with $\det(\Theta(\alpha,\rho)) = -1$. 
Given a sequence of Verblunsky pairs $(\alpha_n, \rho_n)_{n\in\bbZ}$, we define block-diagonal operators $\mathcal{L}$ and $\mathcal{M}$ on $\ell^2(\mathbb{Z})$ via the even and odd subsequences:
\begin{align}
	\mathcal{L} &= \mathcal{L}((\alpha_{2n}, \rho_{2n})_{n\in\mathbb{Z}}) = \bigoplus_{n\in\mathbb{Z}} \Theta(\alpha_{2n}, \rho_{2n}),  \label{def:L}\\
	\mathcal{M} &= \mathcal{M}((\alpha_{2n+1}, \rho_{2n+1})_{n\in\mathbb{Z}}) = \bigoplus_{n\in\mathbb{Z}} \Theta(\alpha_{2n+1}, \rho_{2n+1}),\label{def:M}
\end{align}
where each $\Theta(\alpha_j, \rho_j)$ acts on $\ell^2(\{j, j+1\})$.
The product of these operators defines a \emph{generalized extended CMV (GECMV) matrix} \cite{canteroFivediagonalMatricesZeros2003,BHJ2003}:
\begin{equation}
	\mathcal{E}=\mathcal{E}\big((\alpha_n, \rho_n)_{n\in\mathbb{Z}}\big) = \mathcal{L}\mathcal{M}.
\end{equation}
In the standard basis $\{\delta_n, n\in\bbZ\}$ of $\ell^2(\mathbb{Z})$, $\mathcal{E}$ has the form
\begin{equation} \label{eq:gecmv}
	\mathcal{E} =
	\begin{bmatrix}
		\ddots & \ddots & \ddots & \ddots &&&&  \\
		& \overline{\alpha_0}\rho_{-1} & \boxed{-\overline{\alpha_0}\alpha_{-1}} & \overline{\alpha_1}\rho_0 & \rho_1\rho_0 &&&  \\
		& \overline{\rho_0}\rho_{-1} & -\overline{\rho_0}\alpha_{-1} & {-\overline{\alpha_1}\alpha_0} & -\rho_1\alpha_0 &&&  \\
		&&  & \overline{\alpha_2}\rho_1 & -\overline{\alpha_2}\alpha_1 & \overline{\alpha_3}\rho_2 & \rho_3\rho_2 & \\
		&& & \overline{\rho_2}\rho_1 & -\overline{\rho_2}\alpha_1 & -\overline{\alpha_3}\alpha_2 & -\rho_3\alpha_2 &    \\
		&& && \ddots & \ddots & \ddots & \ddots &
	\end{bmatrix},
\end{equation}
where the boxed entry is $\langle\delta_0,\mathcal{E}\delta_0\rangle$.

The correspondence between GECMV matrices and split-step walks follows from grouping elements of $\bbZ$ into cells of two \cite{CGMV2012QIP,CFO1,CFLOZ}. 
More explicitly, identifying $\ell^2(\mathbb{Z})$ with $\ell^2(\mathbb{Z})\otimes \bbC^2$ via the mapping (see the ordering \eqref{def:ordering})
\begin{equation}\label{eq:base_identification}
	\delta_{2n-1}\mapsto \delta_{n}^+ \quad \text{ and } \quad \delta_{2n}\mapsto \delta_{n}^-,
\end{equation}
which gives the following relationship between the defining objects, the $\Theta$-matrices and the coin operators:
\begin{lemma}\label{CMV-SS}
	Let $(\alpha_n,\rho_n)$, $n\in\bbZ$ be a sequence of Verblunsky pairs and define $\Theta(\alpha_n,\rho_n)$ via \eqref{def:Theta}. 
	Then, identifying $\ell^2(\bbZ)$ with $\ell^2(\bbZ)\otimes \bbC^2$ via the mapping \eqref{eq:base_identification}, the GECMV matrix $\mathcal E=\mathcal L\mathcal M$ can be interpreted as a split-step walk as
	\begin{equation}\label{eq:EW}
		\mathcal{E} = S_+ C_1 S_- C_2, \quad \text{where}\quad C_1(n) = \sigma_1 \Theta(\alpha_{2n}, \rho_{2n}),\quad C_2(n) = \sigma_1 \Theta(\alpha_{2n-1}, \rho_{2n-1}),
	\end{equation}
	with $S_\pm$ as in \eqref{eq:shift_pm} and $\sigma_1$ is the first Pauli matrix acting on the $\bbC^2$-cells.
\end{lemma}
\begin{proof}
	First, note that in terms of decomposition $\ell^2(\bbZ)=\bigoplus_{n\in\bbZ}\ell^2(\{2n-1,2n\})$, $\mathcal{L}$ and $\mathcal{M}$ are given as 
	\begin{equation}\label{eq:LM:shift}
		\mathcal{L}=T \left(\bigoplus_{n\in\bbZ} \Theta(\alpha_{2n},\rho_{2n})\right) T^{-1}, \qquad \mathcal{M}=\bigoplus_{n\in\bbZ} \Theta(\alpha_{2n-1},\rho_{2n-1}),
	\end{equation}
	where $T: \delta_n\mapsto\delta_{n+1}$ is the shift on $\ell^2(\bbZ)$. That is, $\mathcal{L}$ is understood as a shifted block diagonal matrix. Importantly, due to the conjugation with $T$, $\Theta(\alpha_{2n},\rho_{2n})$ and $\Theta(\alpha_{2n-1},\rho_{2n-1})$ in \eqref{eq:LM:shift} now both act on $\ell^2(\{2n-1,2n\})$.
	
	Identifying bases as in \eqref{eq:base_identification}, the shift $T^{\pm 1}$ on $\ell^2(\bbZ)$ acts on $\ell^2(\bbZ)\otimes \mathbb{C}^2$ as $T^{\pm 1}\otimes  P_\pm\sigma_1+\idty_\bbZ\otimes P_\mp\sigma_1$, and thus differs from $S_\pm$ in \eqref{eq:shift_pm} by a flipping operation on the coin space: we have
	\begin{equation}
		T^{\pm 1}(\idty\otimes\sigma_1)=S_\pm.
	\end{equation}
	Plugging this into \eqref{eq:LM:shift} yields
	\begin{align}
		\mathcal{E}&= \mathcal{L}\mathcal{M}= S_+ \left(\bigoplus_{n\in\bbZ} \sigma_1\Theta(\alpha_{2n},\rho_{2n})\right) S_- \left(\bigoplus_{n\in\bbZ} \sigma_1\Theta(\alpha_{2n-1},\rho_{2n-1})\right)=S_+C_1S_-C_2
	\end{align}
	with $C_1,C_2$ as in \eqref{eq:EW}. 
\end{proof}

In the special case where either all even or all odd Verblunsky coefficients vanish, e.g., $\alpha_{2n}= 0$ for all $n$, the GECMV matrix reduces to a shift-coin walk.

Thus, our setting directly translates to that of GECMV matrices with fixed Verblunsky pair $(\alpha, \rho)$: we have $U=SC=\mathcal{L}((0, 1)_{n\in\mathbb{Z}})\,\mathcal{M}((\alpha, \rho)_{n\in\mathbb{Z}})$ and $W=S_+ C_1 S_- C_2=\mathcal{L}((\alpha, \rho)_{n\in\mathbb{Z}})\,\mathcal{M}((\alpha, \rho)_{n\in\mathbb{Z}})$, and electric fields are incorporated analogously to \eqref{eq:electrify}.
Hence, the results of Theorem \ref{thm:speed_limit_W} on maximal velocity, revival relations, and spectral properties established for the electric quantum walks in Theorem \ref{thm:speed_limit_W} carry over directly to the GECMV matrix setting, where the role of the coin parameter $a$ is played by the Verblunsky parameter $\rho$, and the explicit formulas retain the same structure.

A word of warning is in order: the electrified $U_\Phi$ and $W_\Phi$ do not relate to GECMV matrices in the above sense since their Verblunsky pairs are not in $\bbS^3$. Nevertheless, the CGMV connection extends to this setting and yields a standard CMV matrix, see \cite{locQuasiPer} resp. \cite[Appendix B]{CFLOZ}.

\subsection*{Acknowledgements}
H.~A.-R.\ was supported in part by the UAE University under grant number G00004622.

\bibliographystyle{abbrvArXiv}
\bibliography{electric_speed_limit}

\end{document}